\documentclass[journal]{IEEEtran}
\usepackage{amsthm,amssymb,amsmath}
\usepackage{graphicx}
\usepackage{subfigure}
\usepackage{cite}
\usepackage{enumerate}
\usepackage{enumitem}
\usepackage{multirow, multicol}
\usepackage[english]{babel}
\usepackage{subfig}
\usepackage{tikz}
\usetikzlibrary{shapes,arrows,calc,positioning,arrows.meta}
\usepackage{nomencl}
\usepackage{cleveref}
\usepackage{upgreek}
\usepackage{mathrsfs}
\usepackage{upgreek}
\usepackage{mathtools, cuted}
\usepackage{flushend}
\usepackage{algorithm}
\usepackage{algpseudocode}

\newtheorem{theorem}{Theorem}{}
\newtheorem{lemma}{Lemma}{}
{}
\newtheorem{remark}{Remark}{}

\begin{document}
\tikzset{
block/.style = {draw, fill=white!20, rectangle, thin, minimum height=2.0em, minimum width=1.50em},
tmp/.style  = {coordinate}, 
sum/.style= {draw, fill=white!20, circle, node distance=.5cm},
tip/.style = {->, >=stealth', thin, dashed},
input/.style = {coordinate},
output/.style= {coordinate},
pinstyle/.style = {pin edge={to-,thin,black}},
startstop/.style = {draw, rounded rectangle, text centered, draw=black,thick},
io/.style = {trapezium, trapezium left angle=70, trapezium right angle=110, text centered},
process/.style = {rectangle, text centered, draw=black,thick},
decision/.style = {diamond, text centered, draw=black,thick},
arrow/.style = {-{Stealth[scale=1.2]},rounded corners,thick}
}
\title{Scaled Small-Gain Approach to Robust Control of LPV Systems with Uncertain Varying Delay*}
%
%
%

\author{
        Shahin~Tasoujian,
        Saeed~Salavati,
        Karolos~Grigoriadis,
        and~Matthew~Franchek
\thanks{*This paper is a preprint of a paper submitted to the IEEE Conference on Decision and Control 2020}

\thanks{S. Tasoujian, S. Salavati, K. Grigoriadis, and M. Franchek are with the Dept. of Mechanical Engineering at the University of Houston, Houston, Texas, USA, 77004 (e-mail: \{\tt\small  stasoujian@uh.edu).}}

\markboth{submitted to the IEEE Conference on Decision and Control 2020}%
{Shell \MakeLowercase{\textit{et al.}}: Bare Demo of IEEEtran.cls for IEEE Journals}
%



\maketitle

\begin{abstract}
Linear parameter-varying (LPV) systems with uncertainty in time-varying delays are subject to performance degradation and instability. In this line, we investigate the stability of such systems invoking an input-output stability approach. By considering explicit bounds on the delay rate and time-varying delay uncertainty, the scaled small-gain theorem is adopted to form an interconnected time-delay LPV system with input and output vectors of the auxiliary system introduced for the uncertain dynamics. For such an interconnected time-delay LPV system subject to external disturbances, a Lyapunov-Krasovskii functional (LKF) is constructed whose derivative is augmented with the terms resulted from the descriptor method. Then, stability conditions and a prescribed induced $\mathcal{L}_2$-norm in terms of the disturbance rejection performance are derived in a convex linear matrix inequalities (LMIs) setting. Subsequently, a congruent transformation enables us to compute a gain-scheduled state-feedback controller for a class of LPV systems with an uncertain time-varying delay. As a benchmark, we examine the automated mean arterial blood pressure (MAP) control in an individual with hypotension where the MAP response dynamics to drug infusion is characterized in a time-delay LPV representation. Finally, the closed-loop simulation results are provided to demonstrate the provided methodology's performance.
\end{abstract}

%
\IEEEpeerreviewmaketitle

\section{Introduction}
\label{sec:intro}

Model-based control design methodologies provide a systematic and practical framework for addressing stability and performance challenges in a wide variety of real-world control applications. The mathematical modeling of such applications can be carried out via either the certainty equivalence principle and physics laws or identification techniques, which approximates the model in terms of bias and variance error on an identified model \cite{Hou2013model}. However, due to unmodeled and hidden dynamics and problems like aging and external excursions, these model-based approaches suffer from mathematical and real system model mismatches. Consequently, robust control analysis has been introduced as an effective way of dealing with such discrepancy issues \cite{liu2016robust}. However, for varying uncertain parameters in such systems, classical robust methods may introduce an extra conservatism.
	
Practical systems are mainly affected by the presence of delay in their dynamics, which typically leads to performance degradation \cite{Zhang2018analysis}. Moreover, delay uncertainty and, in particular, time-varying delays further pose a robustness challenge in the control of such systems. It is noteworthy that in studying stability analysis of systems with uncertain delays, typically, the delay is assumed to be the sum of a nominal delay and a perturbed uncertain part where the system with the nominal delay is regarded to be asymptotically stable. In this regard, necessary stability conditions for linear time-invariant (LTI) systems with an uncertain constant delay via a frequency-domain approach and sufficient condition for such systems with an uncertain time-varying delay via a Lyapunov-Krasovskii functional (LKF), with a prescribed derivative, has been investigated in \cite{kharitonov2003stability}. The introduced LKF also does not explicitly depend on the bounds of the uncertainties. In \cite{fridman2006stability}, the author has used a complete LKF, with a particular functional form, consisting of a nominal plus additional terms where the former analyzes the system under the nominal delay, and the latter deals with delay perturbations and vanishes as the perturbations disappear. Inspired by this research, the static state-feedback control of input delay systems with an uncertain delay has been addressed in \cite{velazquez2016robust} in which the derivative of LKF is constructed based on a delay Lyapunov matrix. The work in  \cite{Fridman2014} follows the small-gain theorem augmented with a scaling matrix, which provides a less conservative input-output stability analysis. Moreover, to handle delay uncertainty, a factor depending on the delay rate bound has been introduced, which, to further reduce the conservatism, can be selected based on the bound's nominal value. Combined with the LKF strategy, tractable convex conditions in the form of linear matrix inequalities (LMIs) have been achieved for the stability analysis. \cite{alexandrova2018new} has studied the robustness bounds of linear retarded systems with discrete delays where each delay is comprised of a constant part plus a slowly varying perturbation.

Unlike the LTI systems, linear parameter-varying (LPV) systems have a varying structure and thus can adequately describe nonlinear and highly varying systems in a linear setting \cite{salavati2019reciprocal, tasoujian2020robust}. Similarly, control of time-delay LPV systems with the wealth of linear control techniques has been studied in the literature \cite{briat2014linear, tasoujian2019delay, Salavati20162}. In this manuscript, stability and control synthesis of LPV systems with uncertain arbitrarily varying delay and external disturbances is studied that, to the best of our knowledge, has not been addressed in the time-delay LPV control context.
To this end, a candidate LKF is employed to tackle the stability problem with the use of the less conservative descriptor method \cite{Fridman2014} analysis. The input-output approach proposed by \cite{Fridman2014} is used to address the stability of the interconnected input-output LPV system representation under a varying uncertain delay. Further, the worst-case disturbance amplification of the LPV system with uncertain delay in terms of a prescribed induced $\mathcal{L}_2$-norm performance of the system is examined and presented in a convex LMI framework. Subsequently, by taking advantage of a proper congruent transformation, a static gain-scheduled state-feedback control is designed for such systems to first asymptotically stabilize the closed-loop system and second minimize the sub-optimal closed-loop performance index. 

As a practical application, the present paper addresses the problem of mean arterial blood pressure (MAP) regulation and resuscitation in critical cases involving clinical hypotension. Such a delicate and time-sensitive task requires advanced and precise automated drug delivery strategies, which considerably increases the resuscitation chance while avoids under/over-regulated incidents \cite{cao2020blood}. An LPV representation with a varying delay has been used to describe the MAP response dynamics to drug injection. For simulation purposes, we use nonlinear functions to generate the MAP dynamics model parameters in compliance with clinical observations \cite{tasoujian2019robust}. Simulation results confirm that the utilized approach is capable of tracking a reference MAP signal while rejecting the external disturbances caused by medical interventions and complications altering the MAP characteristics. Moreover, the proposed time-delay LPV methodology tackles the challenge of robustness against the uncertain varying time-delay effectively. 

The notation used in this paper is as follows. $\mathbb{R}$ stands for the set of real numbers, $\mathbb{R}_{+}$ is the set of non-negative real numbers, and $\mathbb{R}^n$ and $\mathbb{R}^{k \times m}$ are given to denote the set of real vectors of dimension $n$ and the set of real $k \times m$ matrices, respectively. $\mathbb{S}^{n}$ and $\mathbb{S}^{n}_{++}$ represent the set of real symmetric and real symmetric positive definite $n \times n$ matrices, respectively. The positive definiteness of the matrix $\mathbf{M}$ is designated $\mathbf{M} \succ \mathbf{0}$. The inverse and transpose of a real matrix $\mathbf{M}$ are presented by $\mathbf{M}^{\text{T}}$ and $\mathbf{M}^{-1}$, respectively.
In a symmetric matrix, the asterisk $\star$ in the $(i,\: j)$ element denotes the transpose of the $(j,\: i)$ element. $\mathcal{C} (J,\: K)$ stands for the set of continuous functions mapping a set $J$ to a set $K$.

This work is structured in the following manner. Section \ref{sec:problemstatement} provides the mathematical description of an LPV system with uncertain delay followed by the input-output stability analysis and control design meeting the performance objectives. In Section \ref{sec:Num Ex Res}, the MAP modeling and characterization is presented and the closed-loop simulation results assessing the implemented proposed LPV control methodology for the automated MAP regulation objective are discussed. Finally, the concluding remarks are given in Section \ref{sec:conclusion}.

\section{Problem Statement}\label{sec:problemstatement}

\subsection{Stability and $\mathcal{L}_2$-Gain Analysis of LPV Systems with Uncertain Delay} \label{sec:Stability and L_2 analysis}
We consider the following state-space representation of a general LPV system with an uncertain time-varying state delay:

\begin{equation}
\begin{array}{rcl}
 \dot{\mathbf{x}}(t)  & = &  \mathbf{A}(\boldsymbol{\rho}(t)) \mathbf{x}(t)+\mathbf{A}_\tau(\boldsymbol{\rho}(t)) \mathbf{x}\big(t-\tau (\boldsymbol{\rho}(t))\big) \\[0.1cm] 
 & + & \mathbf{B}_1(\boldsymbol{\rho}(t))\mathbf{w}(t) + \mathbf{B}_2(\boldsymbol{\rho}(t))\mathbf{u}(t), \\[0.2cm] 
  \mathbf{z}(t) & = &\mathbf{C}_1(\boldsymbol{\rho}(t)) \mathbf{x}(t) + \mathbf{C}_{1, \tau}(\boldsymbol{\rho}(t)) \mathbf{x}\big(t-\tau (\boldsymbol{\rho}(t))\big)\\[0.1cm] 
  & + & \mathbf{D}_{11}(\boldsymbol{\rho}(t)) \mathbf{w}(t)+ \mathbf{D}_{12}(\boldsymbol{\rho}(t)) \mathbf{u}(t),\\[0.2cm]
 \mathbf{x}(t_0 + \theta)  & = & \boldsymbol{\phi}(\theta), \:\:\:\: \forall \theta \in [-\overline{\tau}, \: \: 0],
\end{array}
\label{LPVsystem}
\end{equation}
where $\mathbf{x}(t) \in \mathbb{R}^n$ denotes the state vector of the system, $\mathbf{w}(t) \in \mathbb{R}^{n_w}$ is the exogenous input vector with bounded $\mathcal{L}_2$-norm, $\mathbf{u}(t) \in \mathbb{R}^{n_u}$ is the control input vector, $\mathbf{z}(t) \in \mathbb{R}^{n_z}$ is the vector of controlled outputs,
and the matrix coefficients $\mathbf{A}(\cdot)$, $\mathbf{A}_\tau(\cdot)$, $\mathbf{B}_1(\cdot)$, $\mathbf{B}_2(\cdot)$, $\mathbf{C}_1(\cdot)$, $\mathbf{C}_{1, \tau}(\cdot)$, $\mathbf{D}_{11}(\cdot)$, and $\mathbf{D}_{12}(\cdot)$  
are real-valued matrices which are  continuous functions of the time-varying parameter vector $\boldsymbol{\rho}(\cdot) \in \mathscr{F}^\nu _\mathscr{P}$. The scheduling parameter vector is assumed to be measurable in real-time whose trajectories and rate belong to the set $\mathscr{F}^\nu _\mathscr{P}$ defined as
\begin{align}
 \mathscr{F}^\nu _\mathscr{P} \triangleq \{\boldsymbol{\rho}(t) \in \mathcal{C}(\mathbb{R}_{+},\mathbb{R}^{n_s}):\boldsymbol{\rho}(t) \in \mathscr{P}, |\dot{\rho}_i (t)| \leq \nu_i,\nonumber\\  i=1,2,\dots,n_s\},
 \label{eq:parametertraj}
\end{align}
where $n_s$ is the number of parameters and  $\mathscr{P}$ is a compact subset of $\mathbb{R}^{n_s}$.  Moreover, in (\ref{LPVsystem}), $\boldsymbol{\phi}(\theta) \in \mathcal{C}([-\overline{\tau} \:\: 0], \mathbb{R}^n)$ is the functional system's initial condition, and  $\tau (\boldsymbol{\rho}(t))$ is a differentiable scalar function representing the parameter-varying uncertain time delay as follows:

\begin{equation}
    \tau (\boldsymbol{\rho}(t)) = \tau_n + \eta(t), \:\: \vert \eta(t) \vert \leq \mu \leq \tau_n,
    \label{eq:uncertaintimedelay}
\end{equation}

\noindent where $\tau_n$ denotes the nominal delay value and the time-varying uncertain part of the delay is bounded by $\mu$. Moreover, the time-varying delay is considered to be dependent on the scheduling parameter vector and lies in the set $\mathscr{T}^{\nu_\tau}$ defined as  
\begin{align}
\mathscr{T}^{\nu_\tau} \triangleq\! \{ \tau (\boldsymbol{\rho}(t))\! \in \mathcal{C}(\mathscr{P},\mathbb{R}_{+}) \!:\! 0 \leq \tau (\cdot) \leq \overline{\tau} < \infty, \dot{\tau}(\cdot) \leq \nu_\tau\}.
\label{eq:delayset}
\end{align}

Considering the time-delay LPV system (\ref{LPVsystem}), with an allowable parameter vector trajectory in $\mathscr{F}^\nu _\mathscr{P}$, and a time-delay in $\mathscr{T}^{\nu_\tau}$, the design objectives are as follows:
	\begin{itemize}
		\item Internal asymptotic stability of the LPV system with an uncertain varying time-delay in the presence of parameter variations, delay uncertainties, and disturbances, and
		\item Minimization of the worst case amplification of the desired output, $\mathbf{z}$, to a nonzero disturbance signal, $\mathbf{w}$, with bounded energy, \textit{i.e.} solving the problem of $\gamma$-suboptimal induced $\mathcal{L}_2$-norm (energy-to-energy gain) of the mapping $\mathbf{T}_{\mathbf{z}\mathbf{w}}: \mathbf{w} \rightarrow \mathbf{z}$  given by
			\begin{equation}
				{\min}\Vert \mathbf{T}_{\mathbf{z}\mathbf{w}}\Vert_{i,2} = {\min} \underset{\boldsymbol{\rho} \in \mathscr{F}^\nu _\mathscr{P}}{\sup} \:\:\: \underset{\Vert \mathbf{w} \Vert_2 \neq 0 ,\mathbf{w} \in \mathcal{L}_2}{\sup}\:\: \frac{\Vert \mathbf{z} \Vert_2}{\Vert \mathbf{w} \Vert_2}<\gamma,
				\label{eq:Performance Index}
			\end{equation}
	\end{itemize}
\noindent where $\gamma$ is a positive scalar.

In order to examine the stability and $\mathcal{L}_2$-gain analysis of the LPV system with an uncertain time-varying delay, we utilize the small-gain theorem. For this purpose, by considering (\ref{eq:uncertaintimedelay}),  we rewrite the delayed state of the system as follows

\begin{equation}
 \mathbf{x}(t-\tau(t)) =\mathbf{x}(t-\tau_n) - \int^{-\tau_n}_{-\tau_n-\eta(t)} \dot{\mathbf{x}}(t+s)ds,   
\end{equation}

\noindent where the time-varying uncertain part of the delayed state is treated as a disturbance and defined as a new feedback signal:

\begin{equation}
    \mathbf{u}_1(t) = (\Delta \mathbf{y}_1)(t) = -\frac{1}{\mu \sqrt{\mathscr{F}(\nu_\tau)}}\int^{-\tau_n}_{-\tau_n-\eta(t)} \!\!\!\!\!\!\mathbf{y}_1(t+s) ds,
    \label{eq:delayunc_opr}
\end{equation}
\noindent where $\mathscr{F}(\nu_\tau)$ is a continuous function of the time-delay rate, $\nu_\tau$, which will be defined later using an extension of the small-gain theorem. By defining a new auxiliary system, $\Delta$, with additional input and output vectors, namely, $\mathbf{y}_1$ and $\mathbf{u}_1$, the overall interconnected feedback system is shown in Fig. \ref{fig:interconnectedsys}. Accordingly, the unforced time-delay LPV system (\ref{LPVsystem}), \textit{i.e.} no control input or $u\equiv0$, is represented as a feedback interconnected system as follows:
\begin{figure}[!t] 
\hspace*{-.12in}
\centering \includegraphics[width=0.8\columnwidth, height=1.4in]{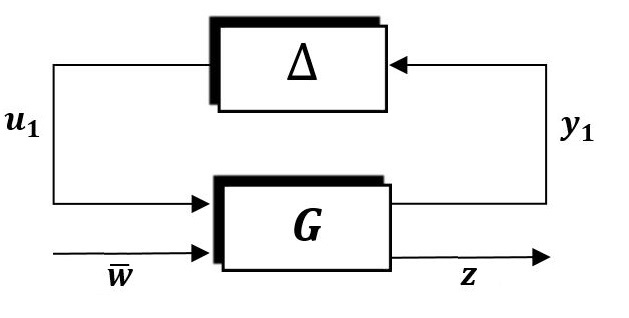}
\caption{The overall interconnected system} 
\label{fig:interconnectedsys}
\end{figure}

\begin{equation}
\begin{array}{rcl}
 \dot{\mathbf{x}}(t)  & = &  \mathbf{A}(\boldsymbol{\rho}(t)) \mathbf{x}(t)+\mathbf{A}_\tau(\boldsymbol{\rho}(t)) \mathbf{x}(t-\tau_n) \\[0.1cm]
 & + & \mu \mathbf{A}_\tau(\boldsymbol{\rho}(t))\mathbb{X}^{-1}\mathbf{u}_1(t) + \gamma^{-1}\mathbf{B}_1(\boldsymbol{\rho}(t))\bar{\mathbf{w}}(t), \\[0.2cm] 
  \mathbf{y}_1(t) & = & \sqrt{\mathscr{F}(\nu_\tau)}\mathbb{X}\dot{\mathbf{x}}(t),\\[0.2cm]
 \mathbf{z}(t) & = &\mathbf{C}_1(\boldsymbol{\rho}(t)) \mathbf{x}(t) + \mathbf{C}_{1, \tau}(\boldsymbol{\rho}(t)) \mathbf{x}(t-\tau_n)\\[0.1cm] 
  & + & \mu \mathbf{C}_{1,\tau}(\boldsymbol{\rho}(t))\mathbb{X}^{-1}\mathbf{u}_1(t)+\gamma^{-1}\mathbf{D}_{11}(\boldsymbol{\rho}(t)) \bar{\mathbf{w}}(t),
\end{array}
\label{LPVsystem_inputoutput}
\end{equation}

\noindent where $\mathbb{X}$ denotes a scaling non-singular matrix and $\bar{\mathbf{w}}(t)=\gamma \mathbf{w}(t)$. 

The following lemma is used to derive the delay-dependent conditions for stability and $\mathcal{L}_2$-gain analysis of the LPV time-delay system with an uncertain time-varying delay.  

\begin{lemma}\label{lem:lemma1} \textbf{(Small-Gain Theorem for Systems with Uncertain Time-Delay)} \cite{Fridman2014}: Considering  $\mathbf{y}_1 = \mathbf{T}_{\mathbf{y}_1\mathbf{u}_1}\mathbf{u}_1$, and $\mathbf{u}_1 = \Delta \mathbf{y}_1$, where both systems $\mathbf{T}_{\mathbf{y}_1\mathbf{u}_1} : \mathcal{L}_2 [0, \infty]  \rightarrow \mathcal{L}_2 [0, \infty]$ and $\Delta : \mathcal{L}_2 [0, \infty]  \rightarrow \mathcal{L}_2 [0, \infty]$ are considered to be input-output stable. The interconnected overall system ($\mathbf{T}_{\mathbf{y}_1\mathbf{u}_1}$, $\Delta$) is input-output stable if $\gamma_0(\Delta)\gamma_0(\mathbf{T}_{\mathbf{y}_1\mathbf{u}_1}) <1$, where $\gamma_0$ is the induced $\mathcal{L}_2$ gain. Moreover, the $\mathcal{L}_2$-gain of the system $\Delta$ is found to be $\gamma_0(\Delta) \leq \overline{\tau} \sqrt{\mathscr{F}(\nu_\tau)}$ where \cite{shustin2007delay}
\begin{equation}
\mathscr{F}(\nu_\tau) = \left\{\begin{array}{lcr}
1, & & -\infty \leq \nu_\tau \leq 1, \\[0.1cm] 
\dfrac{2\nu_\tau-1}{\nu_\tau}, & & 1 < \nu_\tau <2,\\[0.2cm] 
\dfrac{7\nu_\tau-8}{4\nu_\tau-4}, && \nu_\tau \geq2,\\[0.3cm] 
\dfrac{7}{4}, & & \nu_\tau \: is \: unknown, 
\end{array}\right.  
\label{eq:F}
\end{equation}
\noindent with $\Vert \mathbf{T}_{\mathbf{y}_1\mathbf{u}_1}\Vert_{i,2} < \dfrac{1}{\overline{\tau} \sqrt{\mathscr{F}(\nu_\tau)}}$.
\end{lemma}

The following theorem provides the sufficient LMI condition to guarantee the stability and performance objectives:

\begin{theorem}\label{thm:thm1} The unforced LPV system (\ref{LPVsystem}) with an uncertain delay $\vert \eta(t) \vert \leq \mu \leq \tau_n$, over the given sets $\mathscr{F}^\nu _\mathscr{P}$ and $\mathscr{T}^{\nu_\tau}$ is asymptotically stable with the $\gamma$-suboptimal induced $\mathcal{L}_2$-norm, \textit{i.e.} $\Vert \mathbf{z} \Vert_2 \leq \gamma \Vert \mathbf{w} \Vert_2$, if there exist continuously differentiable parameter dependent positive-definite matrix functions $\mathbf{P}(\boldsymbol{\rho}(t))$ , $\mathbf{S}(\boldsymbol{\rho}(t)): \mathscr{F}^\nu_\mathscr{P}\rightarrow\mathbb{S}^{n}_{++}$, positive-definite matrices $\mathbf{Q}$, $\mathbf{R} \in \mathbb{S}^{2n}_{++}$, parameter dependent real matrices $\mathbf{V}_1$, $\mathbf{V}_2$, $\mathbf{V}_3: \mathscr{F}^\nu_\mathscr{P}\rightarrow\mathbb{R}^{n\times n}$, and a positive scalar $\mathbf{\gamma}$ satisfying the following LMI condition
\begin{equation}
    \begin{array}{l}
\left[\begin{array}{cc}
\!\dot{\mathbf{P}}-\!\mathbf{R}+\mathbf{Q}+\!\mathbf{V}_{1}^{\text{T}}\mathbf{A}\!+\!\mathbf{A}^{\text{T}}\mathbf{V}_{1} & \mathbf{P}-\mathbf{V}_{1}^{\text{T}}+\mathbf{A}^{\text{T}}\mathbf{V}_{2} \\[0.1cm]
\star & \!\!\tau_n^2 \mathbf{R}\! +\! \mathscr{F}(\nu_\tau)\mathbf{S}\!-\!\mathbf{V}_2^{\text{T}}\!-\!\mathbf{V}_2 \\[0.1cm]
\star & \star \\[0.1cm]
\star & \star \\[0.1cm]
\star & \star 
\end{array}\right.\\[45pt]
   \left.\begin{array}{cccc}
\mathbf{R}+\mathbf{V}_{1}^{\text{T}}\mathbf{A}_{\tau}+\mathbf{A}^{\text{T}}\mathbf{V}_{3} & \mu\mathbf{V}_{1}^{\text{T}}\mathbf{A}_{\tau}& \mathbf{V}_{1}^{\text{T}}\mathbf{B}_{1} & \mathbf{C}_{1}^{\text{T}} \\[0.1cm]
 \mathbf{V}_{2}^{\text{T}}\mathbf{A}_{\tau}-\mathbf{V_{3}} & \mu\mathbf{V}_{2}^{\text{T}}\mathbf{A}_{\tau} & \mathbf{V}_{2}^{\text{T}}\mathbf{B}_{1} &  \mathbf{0}\\[0.1cm]
-\mathbf{R}-\mathbf{Q}+\mathbf{V}_{3}^{\text{T}}\mathbf{A}_{\tau}+\mathbf{A}_{\tau}^{\text{T}}\mathbf{V}_{3} & \mu\mathbf{V}_{3}^{\text{T}}\mathbf{A}_{\tau}& \mathbf{V}_{3}^{\text{T}}\mathbf{B}_{1} & \mathbf{C}_{1, \tau}^{\text{T}}\\[0.1cm]
\star & -\mathbf{S} & \mathbf{0} & \mu \mathbf{C}_{1, \tau}^{\text{T}}\\[0.1cm]
\star & \star & -\gamma^{2} \mathbf{I} & \mathbf{D}_{11}^{\text{T}}\\[0.1cm]
\star & \star & \star & -\mathbf{I}
\end{array}\!\!\!\right] \!\!\prec\! \mathbf{0}
\end{array}
\label{LMI_unforced}
\end{equation}

\noindent where $\dot{\mathbf{P}}=\sum_{i=1}^{n_s} \dot{\rho}_i(t) \frac{\partial \mathbf{P}(\boldsymbol{\rho}(t))}{\partial \rho_i(t)}$ and the parameter dependence of the matrices is dropped for brevity.

\end{theorem}

\begin{proof}
The proof begins by using the following LKF candidate
\begin{align}
      V (\mathbf{x}_t, \dot{\mathbf{x}_t}, \boldsymbol{\rho}, t)  & = \mathbf{x}^{\text{T}}(t) \mathbf{P}(\boldsymbol{\rho}(t))\mathbf{x}(t) + \int_{t-\tau_n}^{t} \mathbf{x}^{\text{T}}(s) \mathbf{Q} \mathbf{x}(s) ds, \nonumber\\
      & +\int_{-\tau_n}^{0} \int_{t+\theta}^{t} \dot{\mathbf{x}}^{\text{T}}(s) \mathbf{R}_0 \dot{\mathbf{x}}(s) ds d\theta
          \label{eq:lkf1}
\end{align}
\noindent The notation $\mathbf{x}_t(\theta)$ refers to $\mathbf{x}(t+\theta)$ for $\theta \in [\begin{array}{cc}
-\overline{\tau} & 0\end{array}]$ where $\mathbf{x}_{t} \in \mathcal{C} ([-\overline{\tau} \:\: 0], \mathbb{R}^n)$ is the infinite-dimensional delay state vector of the system.
Considering the Lyapunov stability theory, in order to assure the asymptotic stability of the investigated LPV system, we need to evaluate the time derivative of the LKF (\ref{eq:lkf1}) along the trajectories of the LPV system (\ref{LPVsystem}), that is

    \begin{align}
      \dot{V} (\mathbf{x}_t, \dot{\mathbf{x}}_t, \boldsymbol{\rho}, t)  & =  
 2 \dot{\mathbf{x}}^{\text{T}}(t) \mathbf{P}(\boldsymbol{\rho}(t))\mathbf{x}(t) + \mathbf{x}^{\text{T}}(t)\dot{\mathbf{P}} \mathbf{x}(t) \nonumber\\[0.15cm]
   & +\mathbf{x}^{\text{T}}(t) \mathbf{Q}\mathbf{x}(t) +\mathbf{x}^{\text{T}}(t-\tau_n) \mathbf{Q}\mathbf{x}(t-\tau_n)\nonumber\\[0.15cm]
   & +\tau_n \dot{\mathbf{x}}^{\text{T}}(t) \mathbf{R}_0\dot{\mathbf{x}} -  \int_{t-\tau_n}^{t} \!\!\!\!\!\dot{\mathbf{x}}^{\text{T}}(\theta) \mathbf{R}_0 \dot{\mathbf{x}}(\theta) d\theta.
\label{eq:derlyap}
\end{align}
\noindent Employing the Jensen's inequality, the integral term in (\ref{eq:derlyap}) can be upper bounded through
\begin{align}
     - \!\! \int_{t-\tau_n}^{t} \!\!\!\!\!\dot{\mathbf{x}}^{\text{T}}(\theta) \mathbf{R}_0 \dot{\mathbf{x}}(\theta) d\theta \leq\! -\frac{1}{\tau_n}\bigg(\int_{t-\tau_n}^{t} \!\!\!\!\!\!\!\dot{\mathbf{x}}(\theta) d\theta \!\bigg)^{\text{T}}\!\!\mathbf{R}_0\bigg(\int_{t-\tau_n}^{t} \!\!\!\!\!\!\!\dot{\mathbf{x}}(\theta) d\theta \!\bigg)\nonumber \\[0.2cm]
     = -\frac{1}{\tau_n} \big[\mathbf{x}(t) - \mathbf{x}(t-\tau_n) \big]^{\text{T}}\mathbf{R}_0\big[\mathbf{x}(t) - \mathbf{x}(t-\tau_n) \big]. 
     \label{eq:lyapderbound}
\end{align}

Next, in order to derive a relaxed final condition and be able to formulate the final results as an LMI suitable for the synthesis conditions, we use the descriptor technique \cite{Fridman2014}. Introducing three slack variables $\mathbf{V}_1$, $\mathbf{V}_2$, and $\mathbf{V}_3$ and using the LPV system dynamics (\ref{LPVsystem_inputoutput}), we define $\mathcal{I}$ as:

\begin{small}
\begin{equation}
\begin{array}{l}
    \:\:\: \mathcal{I}= \bigg[\mathbf{x}^\text{T}(t)\mathbf{V}^{\text{T}}_1 + \dot{\mathbf{x}}^\text{T}(t)\mathbf{V}^{\text{T}}_2 + \mathbf{x}^\text{T}(t-\tau_n)\mathbf{V}^{\text{T}}_3 \bigg] \nonumber\\[0.4cm]
    \!\bigg(\!\mathbf{A}\mathbf{x}(t) \!+\!  \mathbf{A}_{\tau} \mathbf{x}(t-\tau_n) \!+\!  \mu \mathbf{A}_\tau\mathbb{X}^{-1}\mathbf{u}_1(t) \!+\! \gamma^{-1}\mathbf{B}_1\bar{\mathbf{w}}(t)\!-\! \dot{\mathbf{x}}(t) \!\!\bigg)\! =\! 0.
    \label{eq:I}
\end{array}
\end{equation}
\end{small}

\noindent Considering the augmented forward system with all inputs and outputs, \textit{i.e.} 
$\begin{bmatrix}
\mathbf{y}_1\\
\mathbf{z}
\end{bmatrix}$ = $\mathbf{G}$ $\begin{bmatrix}
\mathbf{u}_1\\
\bar{\mathbf{w}}
\end{bmatrix}$ as in Fig. \ref{fig:interconnectedsys}, the assumption $\vert\vert \mathbf{G} \vert\vert_{i,2}<1$ is equivalent to \cite{Fridman2014} 

\begin{equation}
    \vert\vert \mathbf{y}_1 \vert\vert^2_{\mathcal{L}_2} + \vert\vert \mathbf{z} \vert\vert^2_{\mathcal{L}_2} < \vert\vert \mathbf{u}_1 \vert\vert^2_{\mathcal{L}_2}+ \vert\vert \bar{\mathbf{w}} \vert\vert^2_{\mathcal{L}_2}.
    \label{eq:perfcond}
\end{equation}
Inequality (\ref{eq:perfcond}) satisfies both the condition given in Lemma \ref{lem:lemma1} for the input-output stability of the LPV system with uncertain time-delay, \textit{i.e.} $\Vert \mathbf{T}_{\mathbf{y}_1\mathbf{u}_1}\Vert_{i,2} < \dfrac{1}{\overline{\tau} \sqrt{\mathscr{F}(\nu_\tau)}}$, and also the prescribed performance level given in (\ref{eq:Performance Index}), \textit{i.e.} $\Vert \mathbf{T}_{\mathbf{z}\mathbf{w}}\Vert_{i,2} < \gamma$. Finally, the augmented derivative of the LKF given in (\ref{eq:derlyap}) by the descriptor method's result and (\ref{eq:perfcond}) is 
\begin{align}
     \dot{V} (\mathbf{x}_t,& \dot{\mathbf{x}}_t, \boldsymbol{\rho}, t) +  2\mathcal{I} + \mathbf{y}_1^\text{T}(t)\mathbf{y}_1(t) + \mathbf{z}^\text{T}(t)\mathbf{z}(t) \nonumber \\[0.05cm]
     & - \mathbf{u}_1^\text{T}(t)\mathbf{u}_1(t) - \bar{\mathbf{w}}^\text{T}(t)\bar{\mathbf{w}}(t)  \leq \boldsymbol{\zeta}^\text{T}(t)\boldsymbol{\Omega}\boldsymbol{\zeta}(t) < 0,
     \label{eq:conditionineq}
\end{align}
\noindent where the augmented state vector $\boldsymbol{\zeta}(t)$ is defined as:
\begin{equation}
\boldsymbol{\zeta}^\text{T}(t) \triangleq  
\big[\!\begin{array}{ccccc}
\mathbf{x}^\text{T}(t)  & 
\dot{\mathbf{x}}^\text{T}(t) &
\mathbf{x}^\text{T}(t-\tau_n) &
\mathbb{X}^{-1} \mathbf{u}_1(t) &
\bar{\mathbf{w}}(t)
\end{array}\!\big].
\end{equation}

\noindent Using the bound computed for the integral term in the LKF time-derivative (\ref{eq:lyapderbound}), and substituting the dynamics vectors from (\ref{LPVsystem_inputoutput}) in (\ref{eq:conditionineq}), $\boldsymbol{\Omega}$ is obtained as 

\begin{small}
\begin{equation}
    \begin{array}{l}
    \left[\begin{array}{cc}
\!\!\dot{\mathbf{P}}\!-\!\mathbf{R}+\mathbf{Q}+\!\mathbf{V}_{1}^{\text{T}}\mathbf{A}\!+\!\mathbf{A}^{\text{T}}\mathbf{V}_{1}\! + \mathbf{C}_{1}^{\text{T}} \mathbf{C}_{1}& \!\!\mathbf{P}\!-\mathbf{V}_{1}^{\text{T}}+\mathbf{A}^{\text{T}}\mathbf{V}_{2} \\[0.1cm]
\star & \!\!\!\!\!\tau_n^2 \mathbf{R}\! +\! \mathscr{F}(\nu_\tau)\mathbf{S}\!-\!\mathbf{V}_2^{\text{T}}\!-\!\mathbf{V}_2 \\[0.1cm]
\star & \star \\[0.1cm]
\star & \star \\[0.1cm]
\star & \star 
\end{array}\right.\\[45pt]
 \left.\begin{array}{ccc}
\!\!\!\!\boldsymbol{\Omega}_{13} & \!\!\!\!\!\mu(\mathbf{V}_{1}^{\text{T}}\mathbf{A}_{\tau} + \mathbf{C}_{1}^{\text{T}} \mathbf{C}_{1, \tau})& \!\!\gamma^{-1}(\mathbf{V}_{1}^{\text{T}}\mathbf{B}_{1} + \mathbf{C}_{1}^{\text{T}} \mathbf{D}_{11})  \\[0.15cm]
\!\!\! \mathbf{V}_{2}^{\text{T}}\mathbf{A}_{\tau}-\!\mathbf{V_{3}} & \!\!\!\!\!\!\!\mu\mathbf{V}_{2}^{\text{T}}\mathbf{A}_{\tau} & \!\!\mathbf{V}_{2}^{\text{T}}\mathbf{B}_{1} \\[0.15cm]
\!\!\!\boldsymbol{\Omega}_{33} & \!\!\!\!\!\mu(\mathbf{V}_{3}^{\text{T}}\mathbf{A}_{\tau} + \mathbf{C}_{1, \tau}^{\text{T}} \mathbf{C}_{1, \tau}) &\!\! \gamma^{-1} (\mathbf{V}_{3}^{\text{T}}\mathbf{B}_{1} + \mathbf{C}_{1, \tau}^{\text{T}} \mathbf{D}_{11}) \\[0.15cm]
\!\!\!\star & \!\!\!\!\!\!-\mathbf{S} + \mu^2 \mathbf{C}_{1, \tau}^{\text{T}} \mathbf{C}_{1, \tau} & \!\!\gamma^{-1}\mu\mathbf{C}_{1, \tau}^{\text{T}} \mathbf{D}_{11} \\[0.15cm]
\!\!\!\!\star & \!\!\!\!\!\!\star & \!\!\gamma^{-2} \mathbf{D}_{11}^{\text{T}}\mathbf{D}_{11}-\mathbf{I}
\end{array}\!\!\!\right]\!\!,
\end{array}
\label{eq:LMIinitial}
\end{equation}
\end{small}

\noindent where $\boldsymbol{\Omega}_{13}=\mathbf{R}+\mathbf{V}_{1}^{\text{T}}\mathbf{A}_{\tau}+\mathbf{A}^{\text{T}}\mathbf{V}_{3} + \mathbf{C}_{1}^{\text{T}} \mathbf{C}_{1, \tau}$, $\boldsymbol{\Omega}_{33}=-\mathbf{R}-\mathbf{Q}+\mathbf{V}_{3}^{\text{T}}\mathbf{A}_{\tau}+\mathbf{A}_{\tau}^{\text{T}}\mathbf{V}_{3} + \mathbf{C}_{1, \tau}^{\text{T}} \mathbf{C}_{1, \tau}$, $\mathbf{S}(\boldsymbol{\rho}(t)) = \mathbb{X}^{\text{T}}\mathbb{X}$, and $\mathbf{R}=\dfrac{\mathbf{R}_0}{\tau_n}$. We then pre- and post-multiply (\ref{eq:LMIinitial}) by $diag(\mathbf{I}, \mathbf{I}, \mathbf{I}, \mathbf{I}, \gamma \mathbf{I})$ and its transpose, and apply the Schur complement lemma to (\ref{eq:conditionineq}), to obtain the LMI (\ref{LMI_unforced}) and the proof is accomplished.

\end{proof}

\subsection{State-Feedback LPV Controller Design}
\label{sec:LPVcontrolsubsection}
In this part, we extend the results of Theorem \ref{thm:thm1} for the synthesis of a robust state-feedback gain-scheduling $\mathcal{H}_\infty$ controller for the case of general LPV systems with an uncertain varying time-delay as in (\ref{LPVsystem}). Such a parameter-dependent controller is proposed in the following format:
\begin{equation}
    \mathbf{u}(t) = \mathbf{K}(\boldsymbol{\rho}(t)) \mathbf{x}(t),
    \label{eq:statefeedbackcontrollaw}
\end{equation}
where the controller utilizes full-state information and aims to meet the design objectives as mentioned in Section \ref{sec:Stability and L_2 analysis}. Feeding back the control law (\ref{eq:statefeedbackcontrollaw}) into the LPV system dynamics (\ref{LPVsystem}), the resultant closed-loop system will be
\begin{small}
\begin{equation}
\begin{array}{l}
 \!\!\!\!\dot{\mathbf{x}}(t)\!\! =\!\! \mathbf{A}_{cl}(\boldsymbol{\rho}(t)) \mathbf{x}(t)+\mathbf{A}_\tau(\boldsymbol{\rho}(t)) \mathbf{x}\big(t-\tau (\boldsymbol{\rho}(t))\big) + \mathbf{B}_1(\boldsymbol{\rho}(t))\mathbf{w}(t),\\[0.2cm] 
  \!\!\!\!\mathbf{z}(t) \!\! =\!\! \mathbf{C}_{1,cl}(\boldsymbol{\rho}(t)) \mathbf{x}(t)\! + \! \mathbf{C}_{1, \tau}(\boldsymbol{\rho}(t)) \mathbf{x}\big(t\!-\!\tau (\boldsymbol{\rho}(t))\big)\! +\! \mathbf{D}_{11}(\boldsymbol{\rho}(t)) \mathbf{w}(t),
\end{array}
\label{LPVsystem closed-loop}
\end{equation}
\end{small}

\noindent where $\mathbf{A}_{cl}(\boldsymbol{\rho}(t))=\mathbf{A}(\boldsymbol{\rho}(t))+\mathbf{B}_2(\boldsymbol{\rho}(t))\mathbf{K}(\boldsymbol{\rho}(t))$, $\mathbf{C}_{1, cl}(\boldsymbol{\rho}(t))=\mathbf{C}_1(\boldsymbol{\rho}(t))+\mathbf{D}_{12}(\boldsymbol{\rho}(t))\mathbf{K}(\boldsymbol{\rho}(t))$. By substituting $\mathbf{A}_{cl}(\boldsymbol{\rho}(t))$ and $\mathbf{C}_{1, cl}(\boldsymbol{\rho}(t))$ for $\mathbf{A}$ and $\mathbf{C}_1$ in (\ref{LMI_unforced}), the following theorem presents a sufficient condition for investigating the closed-loop stability and performance with an uncertain delay via such an LPV control design.

\begin{theorem}\label{thm:thm2} There exists a state-feedback gain-scheduling LPV controller (\ref{eq:statefeedbackcontrollaw}), over the sets $\mathscr{F}^\nu _\mathscr{P}$ and $\mathscr{T}^{\nu_\tau}$, to provide the closed-loop system (\ref{LPVsystem closed-loop}) with asymptotic stability and the induced $\mathcal{L}_2$-norm performance index given in (\ref{eq:Performance Index}), if there exist continuously differentiable parameter dependent positive-definite matrix functions $\widetilde{\mathbf{P}}(\boldsymbol{\rho}(t))$ , $\widetilde{\mathbf{S}}(\boldsymbol{\rho}(t)): \mathscr{F}^\nu_\mathscr{P}\rightarrow\mathbb{S}^{n}_{++}$, positive-definite matrices $\widetilde{\mathbf{Q}}$, $\widetilde{\mathbf{R}} \in \mathbb{S}^{2n}_{++}$, parameter dependent real matrix functions $\mathbf{U}(\boldsymbol{\rho}(t)): \mathscr{F}^\nu_\mathscr{P}\rightarrow\mathbb{R}^{n\times n}$, $\mathbf{Y}(\boldsymbol{\rho}(t)): \mathscr{F}^\nu_\mathscr{P}\rightarrow\mathbb{R}^{n_u\times n}$, a positive scalar $\mathbf{\gamma}$, and real scalars $\lambda_2$ and $\lambda_3$ such that the LMI (\ref{eq:LMIClosedloop}) is feasible with $\boldsymbol{\Xi}_{11}=\dot{\widetilde{\mathbf{P}}}-\widetilde{\mathbf{R}}+\widetilde{\mathbf{Q}}+\mathbf{A}\mathbf{U} + \mathbf{U}^{\text{T}}\mathbf{A}^{\text{T}}+ \mathbf{B}_{2} \mathbf{Y} + \mathbf{Y}^{\text{T}}\mathbf{B}_2^{\text{T}}$. Finally, such a control law can then be computed as follows

\begin{equation}
    \mathbf{u}(t) = \mathbf{Y}(\boldsymbol{\rho}(t)) \mathbf{U}^{-1}(\boldsymbol{\rho}(t)) \mathbf{x}(t).
\end{equation}

\begin{figure*}[!t]
\begin{equation}
\!\!\left[\!\!\begin{array}{cccccc}
\boldsymbol{\Xi}_{11}& \widetilde{\mathbf{P}}\!-\mathbf{U}\!+\lambda_2(\mathbf{U}^{\text{T}}\mathbf{A}^{\text{T}}\!+\!\mathbf{Y}^{\text{T}}\mathbf{B}_2^{\text{T}}) & \widetilde{\mathbf{R}} + \mathbf{A}_{\tau}\mathbf{U} + \lambda_3(\mathbf{U}^{\text{T}}\mathbf{A}^{\text{T}}+\mathbf{Y}^{\text{T}}\mathbf{B}_2^{\text{T}}) & \mu \mathbf{A}_{\tau}\mathbf{U} & \mathbf{B}_1 & \mathbf{U}^\text{T}\mathbf{C}_1^\text{T}+\mathbf{Y}^\text{T}\mathbf{D}_{12}^\text{T}\\[0.1cm]
\star & \tau_n^2 \widetilde{\mathbf{R}}\!+\! \mathscr{F}(\nu_\tau)\widetilde{\mathbf{S}}\!-\!\lambda_2(\mathbf{U}\!+\mathbf{U}^\text{T}) & \lambda_2 \mathbf{A}_{\tau}\mathbf{U} - \lambda_3 \mathbf{U}^\text{T} & \lambda_2 \mu \mathbf{A}_{\tau}\mathbf{U} & \lambda_2 \mathbf{B}_1 & \mathbf{0}\\[0.1cm]
\star & \star & -\widetilde{\mathbf{R}}-\widetilde{\mathbf{Q}} + \lambda_3( \mathbf{A}_{\tau}\mathbf{U}+\mathbf{U}^{\text{T}}\mathbf{A}_{\tau}^{\text{T}}) & \lambda_3 \mu \mathbf{A}_{\tau}\mathbf{U} & \lambda_3 \mathbf{B}_1 & \mathbf{U}^\text{T}\mathbf{C}_{1, \tau}^\text{T}\\[0.1cm]
\star & \star & \star & -\widetilde{\mathbf{S}} & \mathbf{0} & \mu\mathbf{U}^\text{T}\mathbf{C}_{1, \tau}^\text{T} \\[0.1cm]
\star & \star & \star & \star  & -\gamma^2 \mathbf{I} & \mathbf{D}_{11}^{\text{T}}\\[0.1cm]
\star & \star & \star & \star & \star & -\mathbf{I}
\end{array}\!\!\!\right]\!\! \prec\! \mathbf{0}
\label{eq:LMIClosedloop}
\end{equation}
\end{figure*}

\end{theorem}

\begin{proof} 
First, we substitute the closed-loop system matrices in the LMI condition (\ref{eq:LMIinitial}) given by Theorem \ref{thm:thm1}, \textit{i.e.} $\mathbf{A}_{cl}$ for $\mathbf{A}$ and $\mathbf{C}_{1, cl}$ for $\mathbf{C}_1$. Next, in order to obtain tractable convex results, we select the slack variables as $\mathbf{V}_1=\mathbf{V}\in \mathbb{R}^{n \times n}$, $\mathbf{V}_2= \lambda_2 \mathbf{V}$, and $\mathbf{V}_3= \lambda_3 \mathbf{V}$ followed by performing a congruent transformation $diag \big(\mathbf{U}^\text{T}, \mathbf{U}^\text{T}, \mathbf{U}^\text{T}, \mathbf{U}^\text{T}, \mathbf{I}, \mathbf{I} \big)$ on (\ref{eq:LMIinitial}). Then, we define the resultant matrix multiplications as $\mathbf{U}^\text{T} \boxdot \mathbf{U} = \widetilde{\boxdot}$ and the new decision variables as $\mathbf{U}=\mathbf{V}^{-1}$ and $\mathbf{Y}=\mathbf{K}\mathbf{U}$ by which the final LMI (\ref{eq:LMIClosedloop}) is obtained and the proof is complete.
\end{proof}

In the next section, we will address the MAP response regulation control problem as a numerical case study to assess the efficiency and robust performance of the proposed gain-scheduling LPV control scheme in several simulation scenarios.

\section{Numerical Example Results and Discussions}\label{sec:Num Ex Res}

\subsection{MAP Response to PHP Drug Dynamics and LPV Modeling}

The investigated application is motivated by the problem of automated MAP regulation using vasoactive drug infusion in critical hypotensive patients resuscitation, such as severe hemorrhage, maternal cesarean hypotension, severe burn, and traumatic brain injury where the inherent feedback loop of body fails to maintain the homeostasis \cite{hollenberg2007vasopressor}. Conventional manual drug administration methods using syringe or infusion pumps to regulate the blood pressure to a desired value, are considered to be labor-intensive, sluggish, and inaccurate, and can lead to fatal consequences and patient's death \cite{cao2020blood}. Thus, designing an advanced control scheme for such a challenging task to automate and computerize the MAP regulation has gained considerable attention, recently \cite{wassar2014automatic, luspay2016adaptive, tasoujian2019robust, tasoujian2019delay, tasoujian2020robust}.

To capture the  MAP response dynamics subject to vasoactive drug infusion such as phenylephrine (PHP), the following LPV model is utilized \cite{tasoujian2020robust,  tasoujian2020real}:

\begin{equation}
\begin{matrix}
 \dot{x}(t) & = & -\dfrac{1}{T(t)} x(t) + \dfrac{K(t)}{T(t)} u(t-\tau(t)),\\[0.4cm]
 y(t)& = & x(t) + d_o(t),\:\:\:\:\:\:\:\:\:\:\:\:\:\:\:\:\:\:\:\:\:\:\:\:\:\:\:\:\:\:\:\:\:
\end{matrix}
\label{eq:MAP state space}
\end{equation}

\noindent where the state variable is considered to capture the MAP variations in $mmHg$ from its baseline value, \textit{i.e.} $x(t) = \Delta MAP (t)= MAP (t) - MAP_b(t)$, $u(t)$ is the drug infusion rate in $ml/h$, $y(t)$ is the patient's measured MAP response output in $mmHg$, $d_o(t)$ denotes the disturbance signal, $K(t)$ characterizes the patient's sensitivity to the drug, $T(t)$ denotes the lag time representing the uptake, distribution, and biotransformation of the drug, and $\tau (t)$ represents the time delay for the drug to reach the circulatory system from the infusion pump.  Fig. \ref{fig:MAPresponse} plots the actual MAP measurements due to a step PHP infusion with a matched response of the utilized model (\ref{eq:MAP state space}). The experimental data has been collected in a swine experiment performed at the Resuscitation Research Laboratory at the University of Texas Medical Branch (UTMB), Galveston, Texas \cite{luspay2016adaptive}. The figure also shows the interpretation of the model parameters $K(t)$, $T(t)$, $\tau(t)$, $MAP_b(t)$ which have been obtained using the least-squares optimization method to fit the actual MAP response measurements. 

\begin{figure}[!t] 
\hspace*{-.1in}
\centering \includegraphics[width=1.05\columnwidth, height=2.10in]{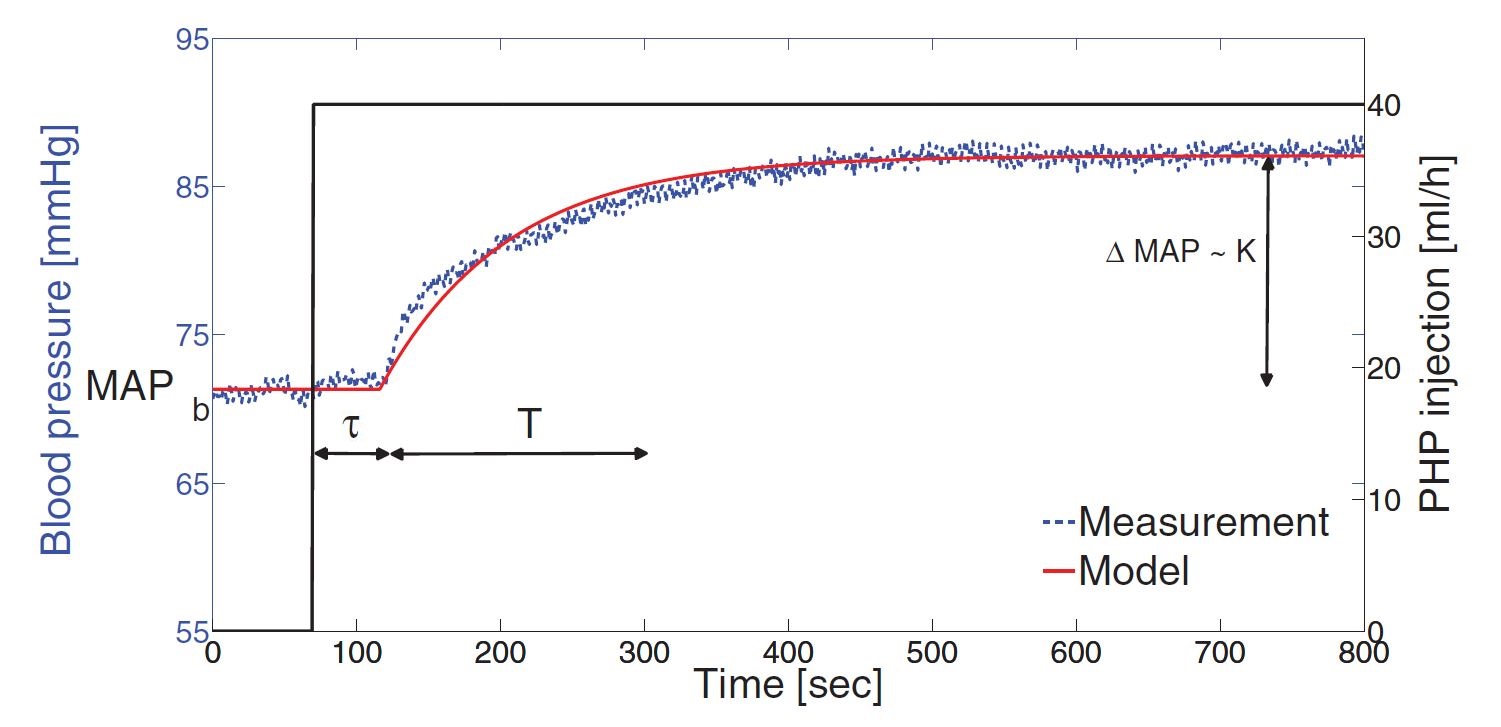} 
\caption{MAP response under step PHP drug infusion} 
\label{fig:MAPresponse}
\end{figure}
In order to utilize the proposed time-delay LPV control design framework introduced in Section \ref{sec:LPVcontrolsubsection}, we need to transform the input delay system (\ref{eq:MAP state space}) into a tractable state-delay LPV representation. For this purpose, we introduce a low-pass input dynamics:
\begin{equation}
u(s)=\frac{\Omega}{s + \Lambda} u_a (s),
\end{equation}
where $\Omega$ and $\Lambda$ are positive scalars that are selected based on the bandwidth of the actuators. Then, the state-space state-delay LPV representation of the MAP response dynamics takes the standard time-delay LPV representation (\ref{LPVsystem})
with the augmented state vector of the system defined as $\mathbf{x}(t):= \mathbf{x}_a (t) = [\begin{array}{ccc} x(t) & u(t) & x_e(t)\end{array}]^{\text{T}}$. $\mathbf{w}(t)=[\begin{array}{cc} r(t) & d_o(t)\end{array}]^{\text{T}}$ stands for the exogenous disturbance vector including the MAP reference command, $r(t)$, and output disturbance signals. $x_e (t)$ is defined for command tracking purposes, \textit{i.e.} $\dot{x}_e (t) = e(t) = r(t)-y(t) = r(t) - (x(t) + d_o(t))$, and thus the state space matrices of the MAP response dynamics in the LPV system representation (\ref{LPVsystem}) are as follows:
\begin{align}
		\!\!\!	& \mathbf{A}(\boldsymbol{\rho}(t)) =\begin{bmatrix}
- \frac{1}{T(t)} & 0 & 0\\ 
0 & -\Lambda & 0 \\ 
-1 & 0 & 0
\end{bmatrix}, \: \mathbf{A}_{\tau}(\boldsymbol{\rho}(t))=\begin{bmatrix}
0 & \frac{K(t)}{T(t)} & 0\\ 
0 & 0 & 0 \\ 
0 & 0 & 0
\end{bmatrix},\nonumber\\
\!\!\!		& \mathbf{B}_1(\boldsymbol{\rho}(t)) \!=\!\!\begin{bmatrix}
0 & \!0\\ 
0 & \!0\\ 
1 & \!-1
\end{bmatrix}\!\!, \mathbf{B}_2(\boldsymbol{\rho}(t))\!=\!\!\begin{bmatrix}
0\\ 
\Omega\\ 
0
\end{bmatrix}\!,  \mathbf{C}_{1}(\boldsymbol{\rho}(t)) \!=\!\! \begin{bmatrix}
0 & \!0 & \!\phi\\
0 & \!0 & \!0
\end{bmatrix}, \nonumber\\[0.1cm]
\!\!\! & \mathbf{C}_{1,\tau}(\boldsymbol{\rho}(t))\!=\!\!\mathbf{0}_{2\times3},  \!\mathbf{D}_{11}(\boldsymbol{\rho}(t)) \!=\!\!\mathbf{0}_{2\times2},  
\mathbf{D}_{12}(\boldsymbol{\rho}(t))\!=\!\!\begin{bmatrix}
0 \\ \psi
\end{bmatrix}\!\!,
\label{eq:matrices1}
\end{align}
\noindent where $\boldsymbol{\rho}(t) = [\begin{array}{ccc} K(t) & T(t) & \tau(t) \end{array}]^{\text{T}}$ denotes the scheduling parameter vector.
 
A major challenge in the precise MAP response regulation control problem is the patient's pharmacological variations, which means that the model parameters and delay, $K(t)$, $T(t)$, $MAP_b(t)$, and $\tau(t)$ could vary significantly from patient-to-patient (inter-patient variability), as well as, for a given patient over time (intra-patient variability) \cite{isaka1993control, rao2003experimental}. Based on clinical observations \cite{tasoujian2020real}, the model parameters variations can be characterized using the following nonlinear functions of the drug injection rate:
\begin{subequations}
	\begin{align}
		& a_k \dot{K}(t) + K(t)  = k_0 exp\{ - k_1 u(t) \} \label{eq:par_K},\\
		& T(t)= sat_{\:[T_{\min}, T_{\max}]} \: \{b_{T} \int_{0}^{t} u(t) \:dt \},\label{eq:par_T}\\
		& \begin{cases}a_{\tau,2} \dddot{\tau}(t) + a_{\tau,1} \ddot{\tau}(t) + \dot{\tau}(t) = b_{\tau,1} \dot{u}(t) + u(t), & \:\:\:\:\:\:\: t\geq t_{i_0}, \nonumber \\\tau(t)=0, & otherwise,\end{cases} \label{eq:par_tau}\\
	\end{align}
\label{eq:nlpatientparameters}\end{subequations}
\noindent where $a_k$, $k_0$, $k_1$, $b_{T}$, $a_{\tau,2}$, $a_{\tau,1}$, and $b_{\tau,1}$ are uniformly distributed random coefficients given in Table \ref{tab:coef} \cite{Craig2004}. Consequently, the parameters variation ranges are $K(t) \in [0.2\;\, 0.8]{mmHg\cdot h}/{ml}$, $T(t) \in [100\;\, 400]sec$, and $\tau(t) \in [0\:\, 70]sec$. Also, the MAP baseline value, $MAP_b(t)$, is assumed to stay at a constant $70\, mmHg$ value. For more details regarding the MAP response dynamics under drug infusion and the real-time model parameters estimation algorithm see \cite{tasoujian2020real} and the references therein.  

\begin{table}[t]
\centering
\caption{Probabilistic distributions of coefficients in (\ref{eq:nlpatientparameters}) }
\label{tab:coef}
\begin{tabular}{cc}
\hline
Parameter    & Distribution                             \\ \hline
$a_k$        & $\mathcal{U}(500, 600)$                  \\
$k_0$        & $\mathcal{U}(0.1, 1)$                    \\
$k_1$        & $\mathcal{U}(0.002, 0.007)$              \\
$b_{T}$      & $\mathcal{U}(10^{-4}, 3 \times 10^{-4})$ \\
$a_{\tau,1}$ & $\mathcal{U}(5, 15)$                     \\
$a_{\tau,2}$ & $\mathcal{U}(5, 15)$                     \\
$b_{\tau,1}$ & $\mathcal{U}(80, 120)$                   \\ \hline
\end{tabular}
\end{table}

\subsection{Automated Closed-loop MAP Regulation Simulation Results and Discussion}

In this part, we have evaluated the potential of the proposed delay-dependent LPV control framework in automated regulation of the MAP response of a simulated nonlinear patient to desired values under different scenarios. Additionally, the closed-loop MAP command tracking results have been compared to the ones of the conventionally implemented fixed-gain PI controller.

For the considered MAP response regulation problem, the controlled output vector, $\mathbf{z}(t)$, in (\ref{LPVsystem}) is defined to be $\mathbf{z}(t) = [\phi \cdot x_e (t) \quad \psi \cdot u (t)]^{\text{T}}$ as (\ref{eq:matrices1}) suggests. The weighting scalars $\phi$ and $\psi$ penalize the tracking error, $x_e(t)$, and the control effort, $u(t)$,  respectively to fulfill desired performance objectives. The gain-scheduled state-feedback controller (\ref{eq:statefeedbackcontrollaw}) has been designed to guarantee the closed-loop asymptotic stability of the LPV time-delay system and to attenuate the worst-case disturbance amplification, \textit{i.e.} minimize the suboptimal induced $\mathcal{L}_2$-norm (or $\mathcal{H}_{\infty}$-norm) of the closed-loop time-delay LPV system (\ref{eq:Performance Index}), over the entire range of the the model parameter trajectories, $\boldsymbol{\rho} \in \mathscr{F}^\nu _\mathscr{P}$, and time-delay variation, $\tau \in \mathscr{T}^{\nu_\tau}$, with the varying time-delay uncertainty lies in the range given in (\ref{eq:uncertaintimedelay}). For this purpose, the results of Theorem \ref{thm:thm2} have been employed to design a robust gain-scheduled LPV controller for calculating the drug injection rate in the automated MAP regulation case study. 

\begin{remark}
The conditions in Theorems \ref{thm:thm1} and \ref{thm:thm2} result in infinite-dimensional convex optimization problems with an infinite number of LMI constraints. To tackle this obstacle, we took advantage of the gridding technique to convert the infinite-dimensional problem to a finite-dimensional convex optimization problem \cite{apkarian1998advanced}. Moreover, a quadratic parameter dependence has been adopted as follows: $\mathbf{M}(\boldsymbol{\rho}(t))=\mathbf{M}_0 + \sum\limits_{i=1}^{n_s}\rho_i(t) \mathbf{M}_{i_1} + \frac{1}{2}\sum\limits_{i=1}^{n_s}\rho_i^2(t) \mathbf{M}_{i_2}$, where $\mathbf{M}(\boldsymbol{\rho}(t))$ represents any of the involved LMI decision matrix variables. Finally, gridding the scheduling parameter space at appropriate intervals leads to a finite set of LMIs to be solved for the unknown matrices and $\gamma$. Also, in order to improve the results, a $2$-dimensional search involving the two scalar variables $\lambda_2$, and $\lambda_3$ is performed to obtain the minimum value of $\gamma$. The MATLAB\textsuperscript{\tiny\textregistered} toolbox YALMIP with Mosek solver is used to solve the corresponding LMI optimization problems \cite{lofberg2004yalmip}.
\end{remark}

First simulation scenario is considered to assess the MAP command tracking performance of the proposed controller in the absence of any disturbance and measurement noise. Such a MAP tracking profile with the associated control effort are plotted in Fig. \ref{fig:tracking_sf} where the results of the introduced gain-scheduled LPV controller have been compared to the performance of a ubiquitously utilized PI controller taken from \cite{wassar2014automatic}. The favorable control task objectives are to regulate the MAP response to desirably track the commanded MAP profile with a minimal settling time and zero steady-state error while keeping the response overshoot within a narrow allowable range. As demonstrated in this plot, the proposed gain-scheduling controller outperforms the fixed design in satisfying the resuscitation objectives. 
\begin{figure}[!t] 
\hspace*{-.12in}
\centering \includegraphics[width=\columnwidth, height=2.1in]{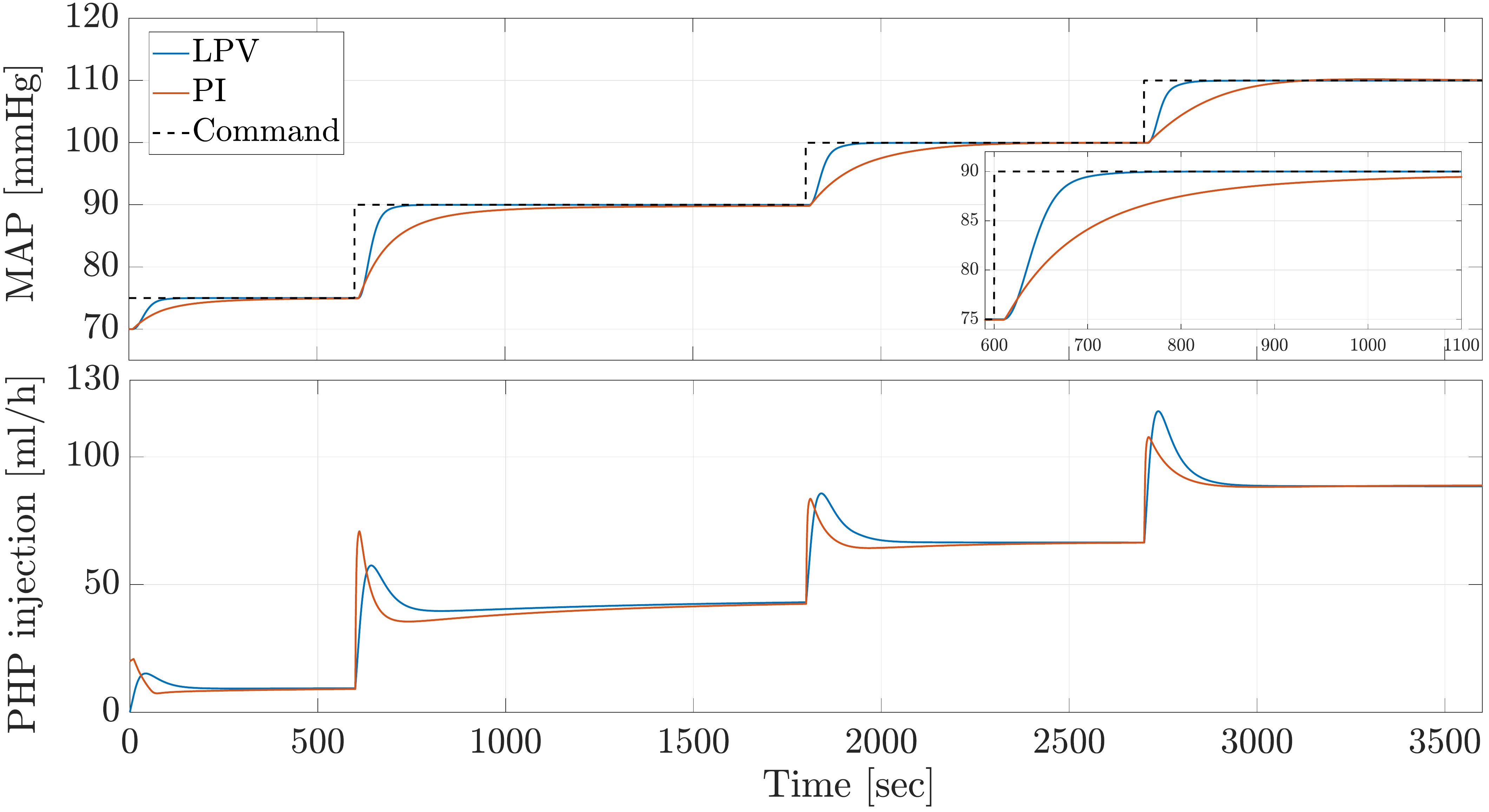}
\caption{MAP response tacking performance and PHP injection rate (control effort) of LPV controller against fixed structure PI controller for disturbance and noise free case} 
\label{fig:tracking_sf}
\end{figure} 

Furthermore, during the MAP regulation process via drug infusion, a patient's MAP response could be influenced by factors other than the vasoactive drug administration such as hemorrhage, unmodeled physiological variations, medications interference like lactated ringers (LR) or sodium nitroprusside (SNP), and any other medical interventions. Fig. \ref{fig:Disturbance} shows a typical profile of such incidents modeled as a disturbance signal. Accordingly, a new scenario has been generated in the simulation environment and Fig. \ref{fig:trackingdistnoise_sf} depicts the performance of the proposed LPV and PI controllers, where the closed-loop system is subject to measurement noise and output disturbances. The considered measurement noise is assumed to be a white noise signal with the intensity of $10^{-3}$. As illustrated, the proposed LPV control method, due to its scheduling structure and robustness in the design, demonstrates a superior MAP command tracking performance with respect to the rise time and speed of the response while desirably rejecting the disturbances and measurement noise.  
\begin{figure}[!t] 
\hspace*{-.12in}
\centering \includegraphics[width=\columnwidth, height=1.85in]{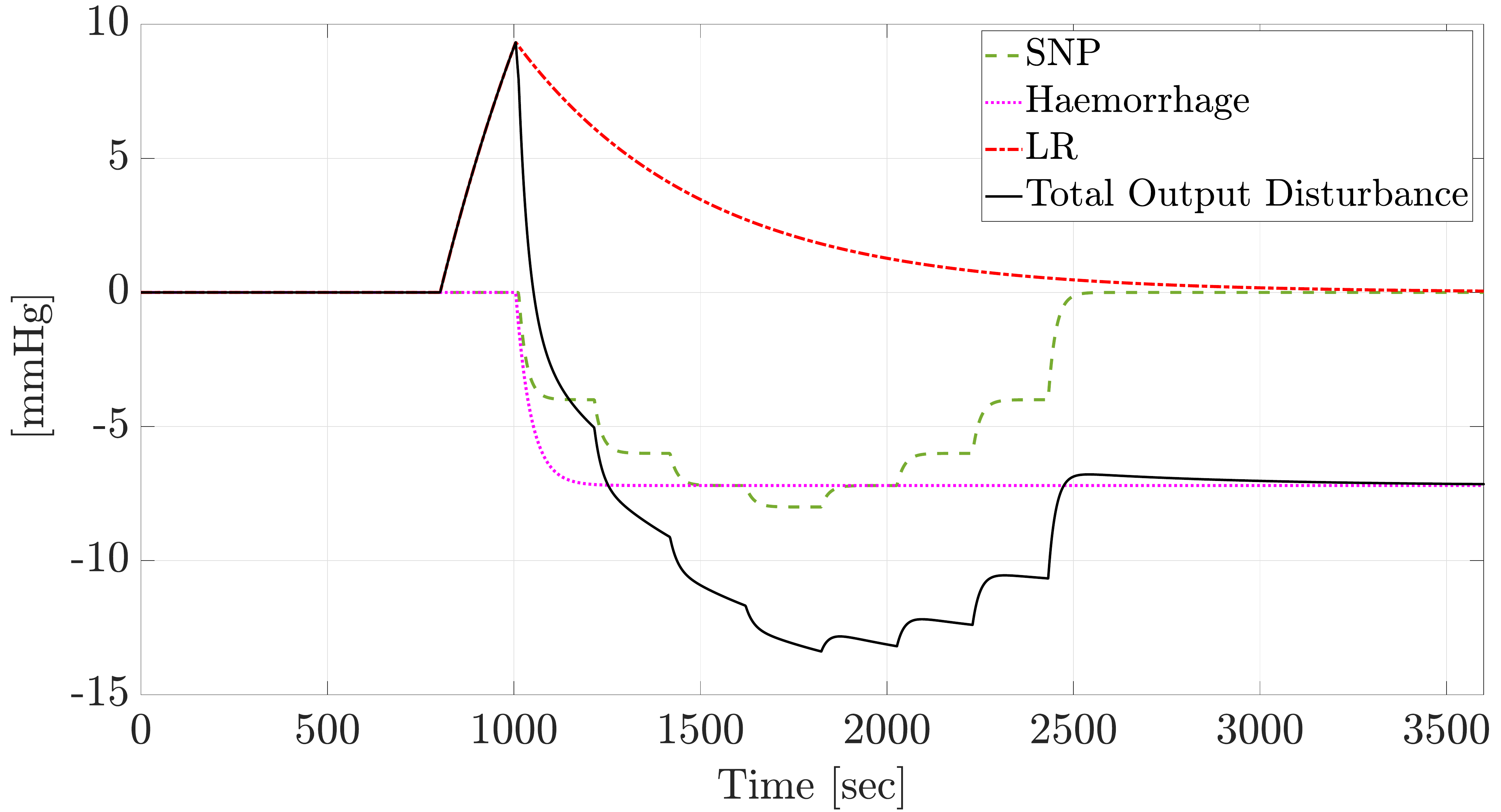}
\caption{Output disturbances profiles} 
\label{fig:Disturbance}
\end{figure}

\begin{figure}[!t]
\includegraphics[width=\columnwidth, height=2.1in]{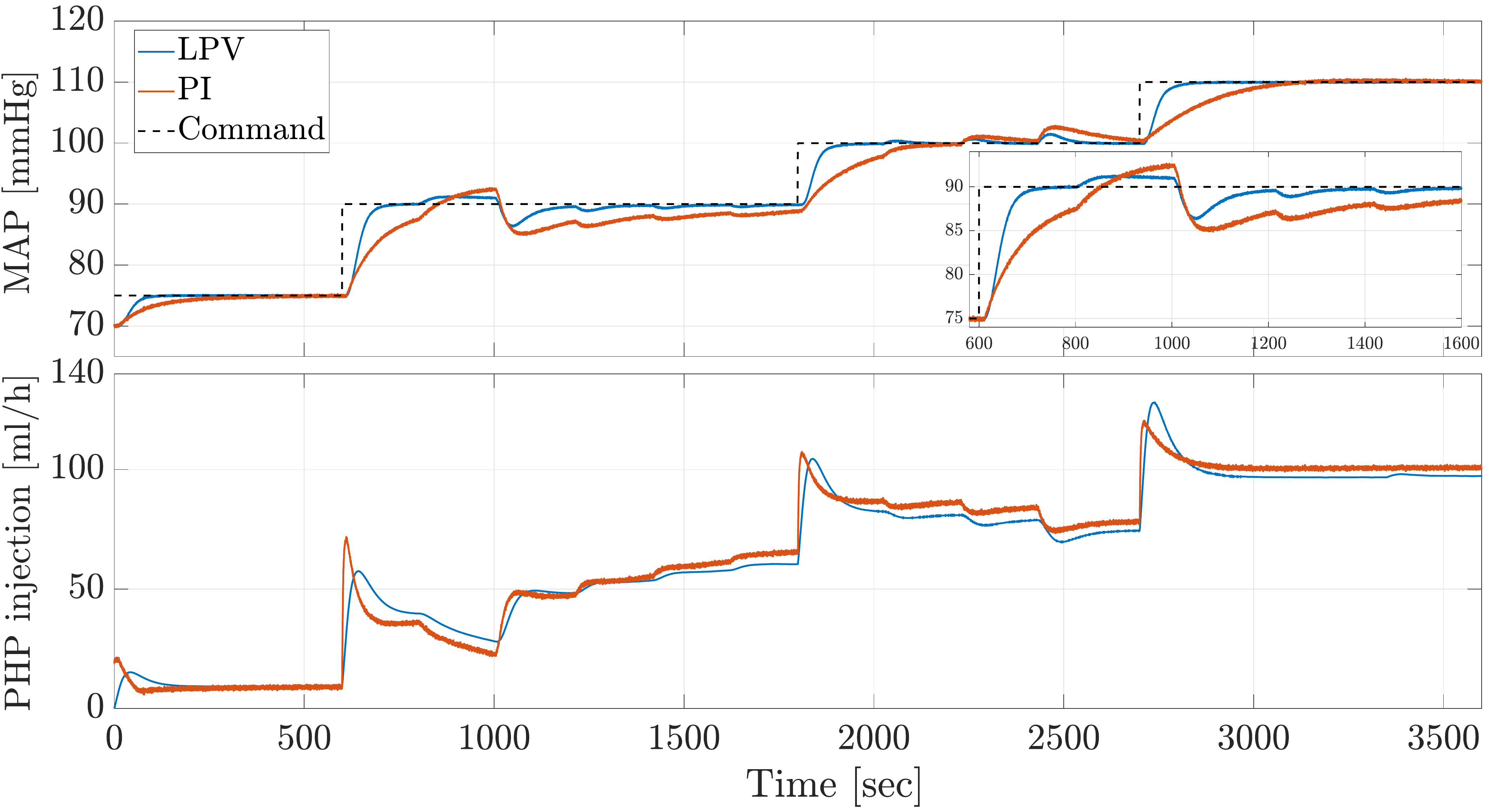}
\caption{MAP response tacking performance and PHP injection rate of LPV controller against fixed structure PI controller under disturbance and measurement noise}
\label{fig:trackingdistnoise_sf}
\end{figure}

Finally, to examine the robustness of the proposed control scheme in handling the time-varying delay uncertainty, we created a scenario in which the model's varying time-delay has been under-estimated by $50 \%$. The closed-loop MAP response of the system with the proposed robust LPV controller has been compared to the response of a fixed-gain PI controller. Fig. \ref{fig:trackingunc_sf} depicts the MAP tracking result of both controllers in this scenario and shows that the PI controller which is designed without considering the time-varying delay uncertainty, demonstrates undesirable oscillatory performance and higher overshoots both in the closed-loop MAP response tracking and also in the PHP injection control input signal. 
\begin{figure}[!t]
\includegraphics[width=\columnwidth, height=2.1in]{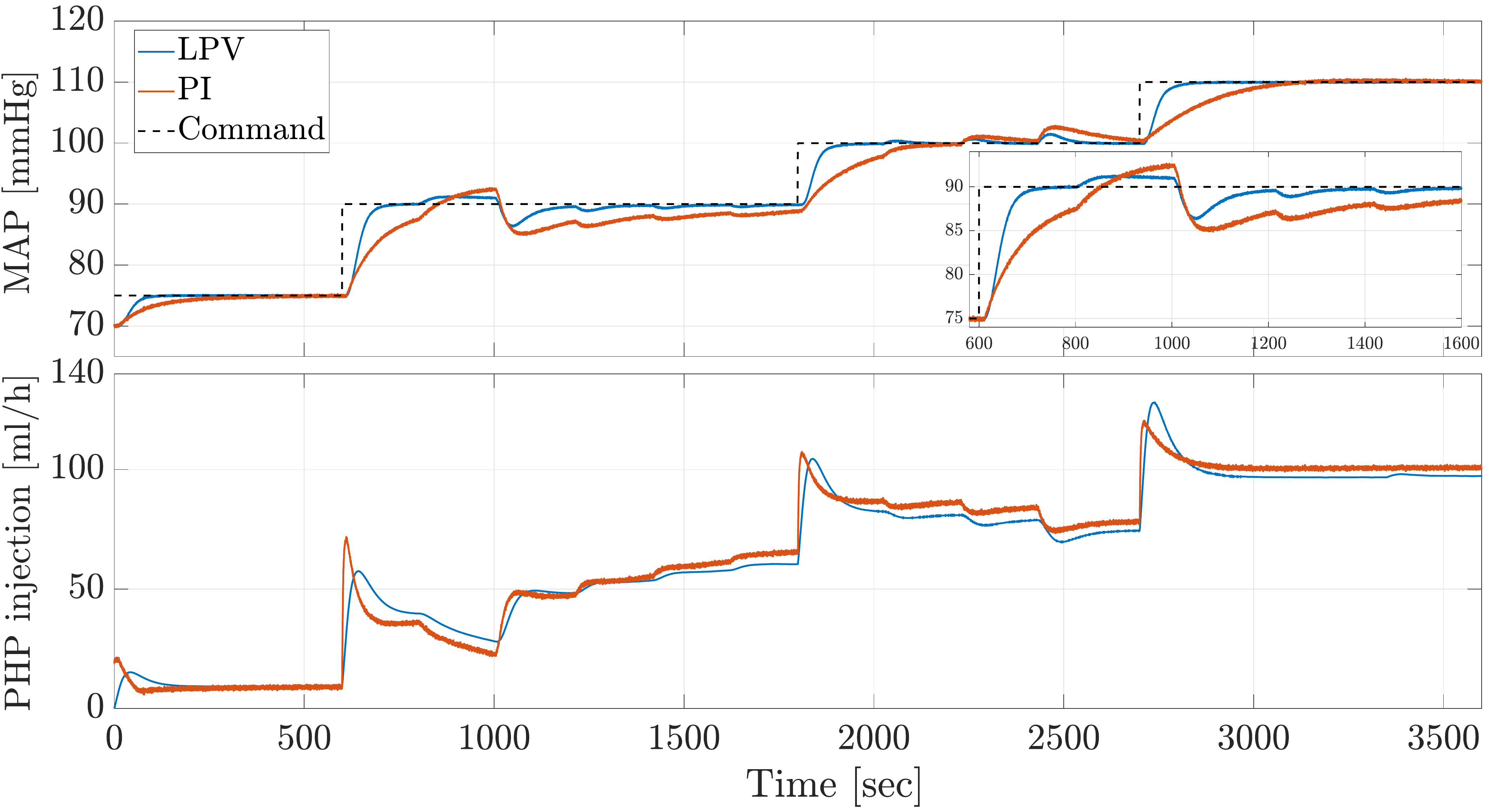}
\caption{MAP response tacking performance and PHP injection rate of LPV controller against fixed structure PI controller under time delay uncertainty}
\label{fig:trackingunc_sf}
\end{figure}
The design parameters are depicted in Table \ref{tab:Parameters}.
\begin{table}[t]
    \caption{Design parameters and performance index}
    \centering
    \begin{tabular}{|c|c|c|c|c|c|c|}
        \hline
        $\lambda_2$ & $\lambda_3$ & $\phi$ & $\psi$ & $\Lambda$ & $\Omega$ & $\gamma$\\
        \hline
         $10.45$ & $-0.55$ & $0.5$ & $1$ & $50$ & $50$ & $33.08$\\
         \hline
    \end{tabular}
    \label{tab:Parameters}
\end{table}
\begin{remark}
Unlike other methods for time delay systems analysis, which handle the varying time-delay uncertainty by considering the largest possible time-delay \cite{lechappe2018prediction}, proposed results based on Theorem \ref{thm:thm2} considers time-delay uncertainty explicitly in the design process, thus, it provides better disturbance attenuation, improved induced $\mathcal{L}_2$-norm performance levels, and less conservative results.
\end{remark}
To conclude, as the results of various cases suggest, the proposed robust gain-scheduled LPV control design is capable of favorably regulating the patient's MAP response to the commanded MAP values while rejecting disturbances, handling model parameter variations, and compensating for time-delay uncertainties.



\section{Conclusion}\label{sec:conclusion}
We have proposed sufficient stability and performance conditions for linear parameter-varying systems with uncertain time-varying delays affected by external disturbances. The uncertain delay has been treated as a nominal delay plus a perturbed function, and in order to confine the perturbation in a stable domain, an input-output stability approach via the small-gain theorem results has been utilized. The sufficient conditions have been formulated in a linear matrix inequality framework using a Lyapunov-Krasovskii functional augmented with the descriptor method approach. Then, control synthesis results have been derived using a proper congruent transformation to provide stability and minimize the disturbance amplification in terms of the induced $\mathcal{L}_2$-norm performance specification of the closed-loop system. The mean arterial blood pressure regulation for critical hypotensive patients via automated drug administration has been studied to assess the performance and effectiveness of the proposed control algorithm. The final closed-loop simulation results have proven the potentials and superiority of the adopted methodology.

\section*{Acknowledgement}\label{sec:Acknow}
Financial support from the National Science Foundation under grant CMMI1437532 is gratefully acknowledged. The collaboration of the Resuscitation Research Laboratory (Dr. G. Kramer) at the University of Texas Medical Branch (UTMB), Galveston, Texas, in providing animal experiment data is gratefully acknowledged.


%



\begin{thebibliography}{13}

\bibitem{Hou2013model}
Z.-S. Hou and Z.~Wang, ``{From model-based control to data-driven control:
  Survey, classification and perspective},'' \emph{Information Sciences}, vol.
  235, pp. 3--35, 2013.

\bibitem{liu2016robust}
K.-Z. Liu and Y.~Yao, \emph{Robust control: theory and applications}.\hskip 1em
  plus 0.5em minus 0.4em\relax Singapore: John Wiley \& Sons, 2016.

\bibitem{Zhang2018analysis}
D.~Zhang and L.~Yu, \emph{Analysis and synthesis of switched time-delay
  systems: the average dwell time approach}, ser. Studies in Systems, Decision
  and Control.\hskip 1em plus 0.5em minus 0.4em\relax Springer Singapore, 2018.

\bibitem{kharitonov2003stability}
V.~L. Kharitonov and S.-I. Niculescu, ``On the stability of linear systems with
  uncertain delay,'' \emph{IEEE Transactions on Automatic Control}, vol.~48,
  no.~1, pp. 127--132, 2003.

\bibitem{fridman2006stability}
E.~Fridman, ``{Stability of systems with uncertain delays: a new "complete"
  Lyapunov-Krasovskii functional},'' \emph{IEEE Transactions on Automatic
  Control}, vol.~51, no.~5, pp. 885--890, 2006.

\bibitem{velazquez2016robust}
J.~E. Vel{\'a}zquez-Vel{\'a}zquez and R.~Galv{\'a}n-Guerra, ``Robust stability
  analysis for linear systems with uncertain fast-varying time delay arising
  from networked control systems.'' \emph{Research in Computing Science}, vol.
  118, pp. 55--64, 2016.

\bibitem{Fridman2014}
E.~Fridman, \emph{Introduction to time-delay systems analysis and
  control}.\hskip 1em plus 0.5em minus 0.4em\relax Basel: Springer, 2014.

\bibitem{alexandrova2018new}
I.~V. Alexandrova and A.~P. Zhabko, ``{A new LKF approach to stability analysis
  of linear systems with uncertain delays},'' \emph{Automatica}, vol.~91, pp.
  173--178, 2018.

\bibitem{salavati2019reciprocal}
S.~Salavati, K.~Grigoriadis, and M.~Franchek, ``{Reciprocal convex approach to
  output-feedback control of uncertain LPV systems with fast-varying input
  delay},'' \emph{International Journal of Robust and Nonlinear Control},
  vol.~29, no.~16, pp. 5744--5764, 2019.

\bibitem{tasoujian2020robust}
S.~Tasoujian, S.~Salavati, M.~Franchek, and K.~Grigoriadis, ``{Robust
  delay-dependent LPV synthesis for blood pressure control with real-time
  Bayesian parameter estimation},'' \emph{IET Control Theory \& Applications},
  2020.

\bibitem{briat2014linear}
C.~Briat, ``Linear parameter-varying and time-delay systems,'' \emph{Analysis,
  Observation, Filtering \& Control}, vol.~3, 2014.

\bibitem{tasoujian2019delay}
S.~Tasoujian, K.~Grigoriadis, and M.~Franchek, ``{Delay-dependent
  output-feedback control for blood pressure regulation using LPV
  techniques},'' in \emph{ASME Dynamic Systems and Control Conference (DSCC)},
  Park City, UT, USA, 2019.

\bibitem{Salavati20162}
S.~Salavati, B.~Ebrahimi, K.~Grigoriadis, and M.~Franchek, ``Reciprocally
  convex feedback controller with feedforward compensation for systems with
  time-varying delay,'' in \emph{ASME Dynamic Systems and Control Conference
  (DSCC)}, Minneapolis, MN, USA, 2016.

\bibitem{cao2020blood}
G.~Cao and K.~Grigoriadis, ``Blood pressure response simulator to vasopressor
  drug infusion (pressorsim),'' \emph{International Journal of Control}, pp.
  1--50, 2020.

\bibitem{tasoujian2019robust}
S.~Tasoujian, S.~Salavati, M.~Franchek, and K.~Grigoriadis, ``{Robust IMC-PID
  and parameter-varying control strategies for automated blood pressure
  regulation},'' \emph{International Journal of Control, Automation and
  Systems}, vol.~17, no.~7, pp. 1803--1813, 2019.

\bibitem{shustin2007delay}
E.~Shustin and E.~Fridman, ``On delay-derivative-dependent stability of systems
  with fast-varying delays,'' \emph{Automatica}, vol.~43, no.~9, pp.
  1649--1655, 2007.

\bibitem{hollenberg2007vasopressor}
S.~M. Hollenberg, ``Vasopressor support in septic shock,'' \emph{Chest}, vol.
  132, no.~5, pp. 1678--1687, 2007.

\bibitem{wassar2014automatic}
T.~Wassar, T.~Luspay, K.~R. Upendar, M.~Moisi, R.~B. Voigt, N.~R. Marques,
  M.~N. Khan, K.~Grigoriadis, M.~Franchek, and G.~C. Kramer, ``Automatic
  control of arterial pressure for hypotensive patients using phenylephrine,''
  \emph{International Journal of Modelling and Simulation}, vol.~34, no.~4, pp.
  187--198, 2014.

\bibitem{luspay2016adaptive}
T.~Luspay and K.~M. Grigoriadis, ``Adaptive parameter estimation of blood
  pressure dynamics subject to vasoactive drug infusion,'' \emph{IEEE
  Transactions on Control Systems Technology}, vol.~24, no.~3, pp. 779--787,
  2015.

\bibitem{tasoujian2020real}
S.~Tasoujian, S.~Salavati, K.~Grigoriadis, and M.~Franchek, ``{Real-time
  cubature Kalman filter parameter estimation of blood pressure response
  characteristics under vasoactive drugs administration},'' in \emph{American
  Control Conference (ACC)}, 2020, pp. 1--8.

\bibitem{isaka1993control}
S.~Isaka and A.~V. Sebald, ``Control strategies for arterial blood pressure
  regulation,'' \emph{IEEE Transactions on Biomedical Engineering}, vol.~40,
  no.~4, pp. 353--363, 1993.

\bibitem{rao2003experimental}
R.~R. Rao, B.~Aufderheide, and B.~W. Bequette, ``Experimental studies on
  multiple-model predictive control for automated regulation of hemodynamic
  variables,'' \emph{IEEE Transactions on Biomedical Engineering}, vol.~50,
  no.~3, pp. 277--288, 2003.

\bibitem{Craig2004}
C.~R. Craig and R.~E. Stitzel, \emph{Modern pharmacology with clinical
  applications}.\hskip 1em plus 0.5em minus 0.4em\relax Baltimore, MD, USA:
  Lippincott Williams \& Wilkins, 2004.

\bibitem{apkarian1998advanced}
P.~Apkarian and R.~J. Adams, ``Advanced gain-scheduling techniques for
  uncertain systems,'' \emph{IEEE Transactions on Control Systems Technology},
  vol.~6, no.~1, pp. 21--32, 1998.

\bibitem{lofberg2004yalmip}
J.~Lofberg, ``{YALMIP: A toolbox for modeling and optimization in MATLAB},'' in
  \emph{IEEE International Conference on Robotics and Automation}.\hskip 1em
  plus 0.5em minus 0.4em\relax IEEE, 2004, pp. 284--289.

\bibitem{lechappe2018prediction}
V.~L{\'e}chapp{\'e}, E.~Moulay, and F.~Plestan, ``{Prediction-based control for
  LTI systems with uncertain time-varying delays and partial state
  knowledge},'' \emph{International Journal of Control}, vol.~91, no.~6, pp.
  1403--1414, 2018.

\end{thebibliography}

\end{document}